\documentclass[10pt,letterpaper]{IEEEtran}
\usepackage{amsmath,amssymb,amsfonts}
\usepackage{graphicx}
\usepackage{fancyvrb}
\usepackage{amsthm}
\usepackage{mathbbol}
\usepackage{ifthen}
\usepackage{cancel}
\usepackage[colorlinks, linkcolor=Blue, filecolor =Red,citecolor=RoyalPurple]{hyperref}
\usepackage[lmargin=0.76in,rmargin=0.76in,tmargin=0.73in,bmargin=0.78in]{geometry}
\usepackage{stackengine}
\usepackage[dvipsnames,usenames]{color}

\newtheorem{lem}{Lemma}

\newtheorem{cor}{Corollary}

\newtheorem{remark}{Remark}

\def\Prob{{\sf P}}

\def\delequalRHS{\mathrel{\ensurestackMath{\stackunder[1pt]{=}{\scriptstyle\triangledown}}}}

\def\textUB{\text{max}}
\def\textLB{\text{min}}

\def\textAgg{{\scriptsize \hbox{agg}}}
\def\textOn{{\footnotesize \hbox{on} }}
\def\textOff{ {\footnotesize \hbox{off} } }
\def\textTS{ {\scriptsize \hbox{TS} } }
\def\textBA{\text{BA}}

\def\tempBin{\ensuremath{\lambda}}
\def\numTCLs{\ensuremath{ {\sf{N_{tcl}} } }}

\def\LockedEll{\ensuremath{{\sf{L} }^{\ell} }}

\newcommand{\version}{arxiv}
\newcommand{\markedManu}{MARKED}

\newboolean{showArxiv}
\newboolean{showArxivAlt}

\ifthenelse{\equal{\version}{journal}}
{
	\setboolean{showArxiv}{false} 
	\setboolean{showArxivAlt}{true}
}
{
}

\ifthenelse{\equal{\version}{arxiv}}
{
	\setboolean{showArxiv}{true} 
	\setboolean{showArxivAlt}{false}
}
{
}

\ifthenelse{\equal{\markedManu}{MARKED}}
{

}
{
}

\ifthenelse{\equal{\markedManu}{nonMARKED}}
{

}
{
}

\begin{document}
\title{\vspace{6.4mm}Control oriented modeling of TCLs}

\author{
	\IEEEauthorblockN{Austin R. Coffman\IEEEauthorrefmark{1}$^,$\IEEEauthorrefmark{3}, Ana Bu\v{s}i\'{c}\IEEEauthorrefmark{2}, and Prabir Barooah\IEEEauthorrefmark{1}}
	\thanks{\IEEEauthorrefmark{1}University of Florida, \IEEEauthorrefmark{2}INRIA} 
	\thanks{\IEEEauthorrefmark{3}corresponding author, email: bubbaroney@ufl.edu.}
	\thanks{AC and PB are with the Dept. of Mechanical and Aerospace Engineering, University of Florida, Gainesville, FL 32601, USA. AB is with Inria Paris and also with DI ENS, \'Ecole Normale Sup\'erieure, CNRS, PSL Research
		University, Paris, France. The research reported here has been partially supported by the NSF through award 1646229 (CPS-ECCS). }
	\vspace{-0.75cm}
}
\maketitle
	\begin{abstract}
		Thermostatically controlled loads (TCLs) have the potential to be a valuable resource for the Balancing Authority (BA) of the future.   Examples of TCLs include household appliances such as air conditioners, water heaters, and refrigerators. Since the rated power of each TCL is on the order of kilowatts, to provide meaningful service for the BA, it is necessary to control large collections of TCLs. To perform design of a distributed coordination/control algorithm, the BA requires a control oriented model that describes the relevant dynamics of an ensemble. Works focusing on solely modeling the ensemble date back to the 1980's, while works focusing on control oriented modeling are more recent. In this work, we contribute to the control oriented modeling literature. We leverage techniques from computational fluid dynamics (CFD) to discretize a pair of Fokker-Planck equations derived in earlier work~\cite{MalhameElectricTAC:1985}. The discretized equations are shown to admit a certain factorization, which makes the developed model useful for control design. In particular, the effects of weather and control are shown to independently effect the system dynamics. 
	\end{abstract}
	
\section{Introduction}	
An envisioned future for the power grid is one that relies more on renewable generation sources. An inevitable challenge in this scenario is the inherent variability present in renewable generation sources, such as solar or wind. This variability requires grid operators to ramp controllable resources up and down to meet the demand when renewable generation does not. Ramp rate constraints prevent conventional generation from handling this mismatch completely. Grid level storage from batteries is expensive. Thus a new resource is being investigated to help fill the mismatch where conventional generators and batteries fall short: flexible loads.

Flexible loads are loads that can vary their power consumption, around a nominal value, without affecting the QoS of the load. Nominal refers to the power consumption without control from the BA, and power deviation as the amount deviated from nominal. The nominal consumption, for example, for air conditioners, is largely determined by ambient weather conditions. Examples of flexible loads include, TCLs~\cite{callaway2011achieving,ChenDistributedIMA:2017,CoffmanVESBuildSys:2018,matkoccal:2013} (e.g., water heaters and air conditioners) pumps for agricultural purposes~\cite{AghajFarm:2019}, pool cleaning~\cite{ChenDistributedIMA:2017}, and heating~\cite{LeeGridJESBC:2020} and HVAC systems in commercial buildings~\cite{haokowlinbarmey:2013}. Since the rated power of some flexible loads is quite small, it is necessary to consider collections of flexible loads. In the following we focus solely on TCLs.

While TCLs are a flexible load, their nominal behavior needs to be altered to take advantage of their flexibility. That is, in order to be utilized as a resource, the BA needs to issue implementable control commands to each TCL that reflects its needs. These inputs modify slightly the nominal behavior of each TCL, so that in aggregate the collection tracks the desired power deviation. Examples of inputs in the current literature include: (i) thermostat set point changes~\cite{callaway2011achieving,BashashModelingTCST:2012}, (ii) randomized control algorithms~\cite{ChenDistributedIMA:2017,CoffmanVESBuildSys:2018}, and (iii) direct load control (for example, the priority stack controller within~\cite{hao_aggregate:2015}).   

From the standpoint of control design, it is also important to have a model that describes the effects of the control input on the ensembles power consumption. Ref.~\cite{MalhameElectricTAC:1985} develops a pair of coupled Fokker-Planck equations to model an ensemble of TCLs during nominal operation. The Fokker-Planck equations are partial differential equations (PDE's) that describe the time evolution of a certain probability density function (pdf). Upon discretization, the coupled PDE's turns into coupled ODE's and the pdf turns into a probability mass function (pmf) that holds similar interpretation as the ``binned'' state common in the literature~\cite{matkoccal:2013,LiuDistributedTIE:2016}. However, since the PDE's are developed to model nominal operation it is, in general, a design choice on how to introduce control into this modeling framework.

In this work, we develop a control oriented framework for ensembles of TCLs.
This framework is based on discretization of the coupled Fokker-Planck PDE's exposed in~\cite{MalhameElectricTAC:1985}. The main contribution is that our discretization allows us to infer a special structure of the resulting discretized system. This structure decomposes the effects of exogenous disturbances, such as weather, and the control input. This structure has the so-called ``conditional independence'' decomposition appearing as an assumption in the work~\cite{busmey:CDC:2016}. There are at least two advantages of the identified structure: (i) it elucidates how one can introduce a control input and (ii) it allows for computationaly efficient control design. To our knowledge, use of discretization to obtain this conditional independence structure is absent from prior literature. 

 
\subsection{Literature review} 
There are two important ingredients for controlling collections of TCLs: (i) identifying a control input and (ii) modeling the effects of the control input. As previously mentioned, many works modify the modeling framework exposed in~\cite{MalhameElectricTAC:1985} to achieve both points (i) and (ii). 

Since PDE's are infinite dimensional, some form of a discretization is required for the eventual purpose of control design. After discretization, a finite dimensional population model can be developed. This model is of the form $\nu_{k+1} = \nu_kP_k$ where $P_k$ is a Markov transition matrix and $\nu_k$ is a marginal distribution. The works~\cite{BenenatiTractableArxiv:2019,NazirAnalysisEPSR:2020,TotuDemandCST:2017} take this route, and $\nu_k$ represents the ``fraction of flexible loads with state value in a certain bin.'' Alternative to discretization, one can define this fractional state vector in an ad-hoc fashion and develop population models by analytically computing transition probabilities~\cite{LiuDistributedTIE:2016}. It is also possible to estimate the population model through measured data~\cite{KaraImpactTSG:2015} or Monte-Carlo simulation~\cite{matkoccal:2013}.

To introduce control to the discretized models, one popular approach is to define a vector control input with $i^{th}$ entry as ``the fraction of TCLs to switch mode state in bin $i$''~\cite{matkoccal:2013,LiuDistributedTIE:2016}, leading to a bilinear control system. Another approach assumes the ability to change the thermostatic set point of each TCL. The effects of this control input can be modeled prior to discretization, and after discretization, like the previous approach, a bilinear control system results~\cite{BashashModelingTCST:2012}. One more approach introduces control by allowing the TCLs mode state to be determined through a randomized control policy~\cite{KaraImpactTSG:2015}. 

In regards to discretization our approach belongs to the first class of methods, i.e., we discretize the pde's to obtain a population model of the form $\nu_{k+1} = \nu_kP_k$. However, to introduce control to this control free population model, our approach is different from much of the literature. We study the structure of $P_k$. Elaborating, the control free population model is based on the TCLs nominal thermostatic policy. Is it then possible to `factor' this policy out, i.e., rewrite the population model as $\nu_{k+1} = \nu_k\Phi_k G_k$ so that an arbitrary control policy can be inserted instead? The answer is affirmative, and this factorization refers to the conditional independence form mentioned prior. Key in identifying this is in how we discretize the set of coupled PDEs.

In numerical experiments we evaluate the fidelity of our discretized model by comparing the state of the model to empirical quantities obtained from a simulation of TCLs. In addition, we also offer a preview of control results using the developed model with the identified structure.   

The paper proceeds as follows. In Section~\ref{sec:ModelDev} the model of the individual TCL is introduced. In Section~\ref{sec:discretzation} the PDEs introduced are discretized and in Section~\ref{sec:indStruc} the structure of the discretized model is identified. Numerical experiments are reported in Section~\ref{sec:aggModel} and we conclude in Section~\ref{sec:conc}. 	
\section{Modeling: Individual TCL} \label{sec:ModelDev}
\subsection{Deterministic Model}
An individual TCL has two state variables: (i) a temperature denoted $x(t)$ and (ii) an on/off mode denoted $m(t)$.
We consider two models for an individual TCL. The first is the following ODE,
\begin{align} \label{eq:detModelTCL}
	\frac{d}{dt} x(t) = f_m(x,t),
\end{align}
where
\begin{align}
	f_m(x,t) = -\frac{1}{RC}\left(x - \theta_a(t)\right) - m(t)\frac{\eta P}{C}.
\end{align}
The rated electrical power consumption is denoted $P$ with coefficient of performance (COP) $\eta$. The parameters $R$ and $C$ denote thermal resistance and capacitance, respectively. The signal $\theta_a(t)$ is the ambient temperature. In the following we identify $m(t) = 1$ and $m(t) =$ on, as well as $m(t) = 0$ and $m(t) =$ off. The nominal power for the TCL is the value of $P$ so that $f_1(\lambda^{\text{set}},t) = 0$. Solving this yields the nominal power for a TCL as,
\begin{align} \label{eq:invTCLbase}
\bar{P}^{\text{ind}}(t) = \frac{\theta_a(t)-\lambda^{\text{set}}}{\eta R}.
\end{align} 

\subsection{Stochastic Model}
The stochastic model is based on the deterministic model. Consider the It\^{o} stochastic differential equation (SDE),
\begin{align} \label{eq:stoModelTCL}
d x(t) = f_m(x,t)dt + \sigma^2dB(t),
\end{align}
where $B(t)$ is Brownian motion with diffusion coefficient $\sigma^2>0$. The quantity $\sigma^2dB(t)$ in~\eqref{eq:stoModelTCL} captures modeling errors in~\eqref{eq:detModelTCL}.
\paragraph{Nominal thermostat policy}
To state the Fokker-planck PDE's as in~\cite{MalhameElectricTAC:1985} we denote the nominal thermostat control policy:
\begin{align} \label{eq:thermoContLaw}
\lim_{\epsilon \rightarrow 0}\ m(t + \epsilon) = \begin{cases}
1, & x(t) \geq \tempBin^{\textUB}. \\
0, & x(t) \leq \tempBin^{\textLB}. \\
m(t), & \text{o.w.} 
\end{cases}
\end{align}
The quantities $\tempBin^{\textUB}$ and $\tempBin^{\textLB}$ respectively set the upper and lower temperature limits (i.e., the thermostatic ``deadband'') for $x(t)$. The midpoint of the deadband interval $[\lambda^{\textLB},\lambda^{\textUB}]$ is denoted $\lambda^{\text{set}}$.
The nominal policy~\eqref{eq:thermoContLaw} is only temporary; in Section~\ref{sec:indStruc} we show how to model the effects of an arbitrary randomized policy.

Now, consider the following marginal pdfs $\mu_{\textOn}, \mu_{\textOff}$:
\begin{align} \label{eq:probOnState}
	\mu_{\textOn}(\tempBin,t)d\tempBin &= \Prob\left((\tempBin < x(t) \leq \tempBin + d\tempBin), \ m(t) = \text{on} \right), \\ \label{eq:probOffState}
	\mu_{\textOff}(\tempBin,t)d\tempBin &= \Prob\left((\tempBin < x(t) \leq \tempBin + d\tempBin), \ m(t) = \text{off} \right),
\end{align}
where $\Prob(\cdot)$ denotes probability, and for now $m(t)$ evolves according to~\eqref{eq:thermoContLaw}. It was shown in~\cite{MalhameElectricTAC:1985} that the densities $\mu_{\textOn}$ and $\mu_{\textOff}$ satisfy the Fokker Planck equations,
\begin{align} \label{eq:pdeOnMode}
	\frac{\partial}{\partial t}\mu_{\textOn}(\lambda,t) &= -\nabla_\lambda\Big(f_{\textOn}(\lambda,t)\mu_{\textOn}(\lambda,t)\Big) +\frac{\sigma^2}{2}\nabla^2_\lambda\mu_{\textOn}(\lambda,t) \\ \label{eq:pdeOffMode}
	\frac{\partial}{\partial t}\mu_{\textOff}(\lambda,t) &= -\nabla_\lambda\big(f_{\textOff}(\lambda,t)\mu_{\textOff}(\lambda,t)\big) +\frac{\sigma^2}{2}\nabla^2_\lambda\mu_{\textOff}(\lambda,t)
\end{align}
that are coupled through their boundary conditions~\cite{MalhameElectricTAC:1985}, which are listed later in Section~\ref{sec:boundCond}.

There are at least two ways that the coupled equations~\eqref{eq:pdeOnMode}-\eqref{eq:pdeOffMode} can be used for modeling: (i) to model a \emph{single} TCL and (ii) to model an \emph{ensemble} of TCLs. That is, for (i) the quantities~\eqref{eq:probOnState}-\eqref{eq:probOffState} represent the \emph{probability} that a single TCLs state resides in the respective interval. For (ii) the quantities~\eqref{eq:probOnState}-\eqref{eq:probOffState} represent the \emph{fraction} of TCLs whose state resides in the respective interval. How the equations~\eqref{eq:pdeOnMode}-\eqref{eq:pdeOffMode} (specifically their discretized form) can be used to model an ensemble is discussed further in Section~\ref{sec:contAggMod}. 

\begin{figure}
	\centering
	\includegraphics[width=1\columnwidth]{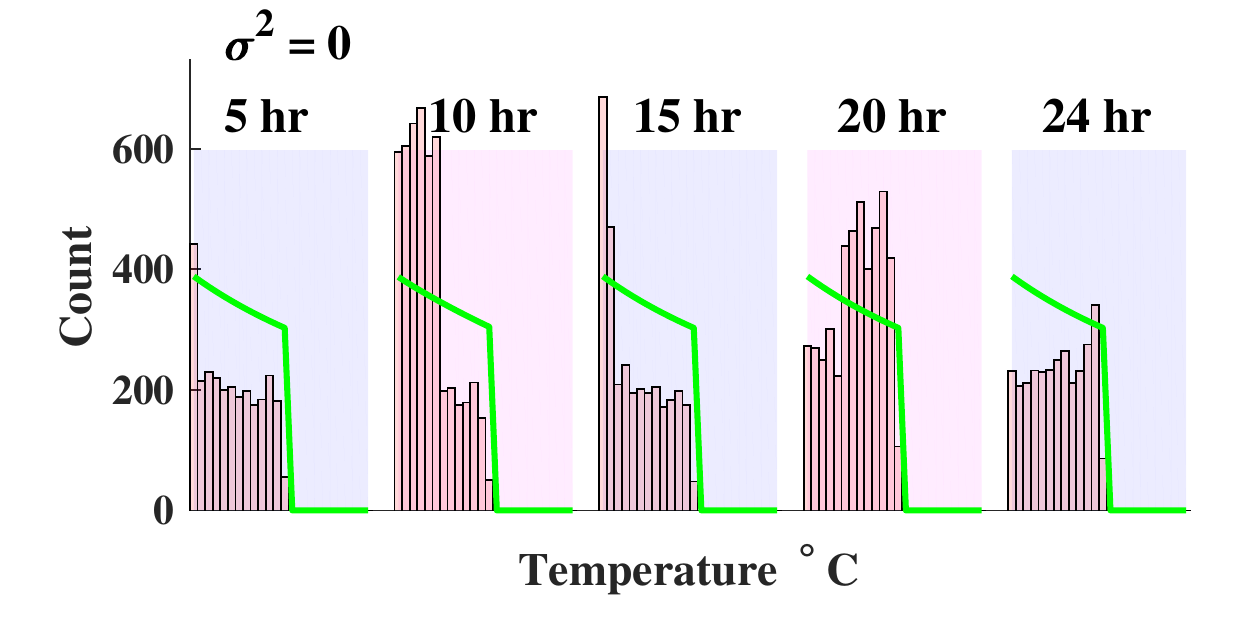}
	\caption{Discrepancy between the state of the advection equation and the histogram of the population for various time steps. Each histogram is over the temperature state for all of the on TCLs at the specified time.}
	\label{fig:advEqPeriodic}
\end{figure}
\subsubsection{Motivation for Stochastic Model}
While transport type arguments can be used to develop a pair of coupled advection equations (equations~\eqref{eq:pdeOnMode}-\eqref{eq:pdeOffMode} with $\sigma^2=0$) for the deterministic model~\cite{BashashModelingTCST:2012}, the state of these advection equations will not agree with the pointwise in time histogram of a population of TCLs simulated with~\eqref{eq:detModelTCL} (see Figure~\ref{fig:advEqPeriodic}). To see why, consider the following: without noise TCLs are periodic whereas discretization of the advection equations yields a Markov transition matrix that is irreducible and aperiodic. Hence, the iteration of this transition matrix will converge to a limiting and invariant distribution, whereas the samples from the TCLs will not since the TCL behavior is periodic. This behavior is shown in Figure~\ref{fig:advEqPeriodic}, the discretized state of the advection equation remains relatively constant while the histogram of the ensemble does not; their is no suggestion of convergence even after 24 hours.

Thus, the stochastic model~\eqref{eq:stoModelTCL} has two advantages: (i) it captures modeling errors and heterogeneity the deterministic model~\eqref{eq:detModelTCL} can not, and (ii) it also guarantees a correspondence between simulation samples from~\eqref{eq:stoModelTCL} and the state of the coupled PDEs~\eqref{eq:pdeOnMode}-\eqref{eq:pdeOffMode} (see Figure~\ref{fig:histSimCompare}).

\subsubsection{Forward thinking motivation}Further, the PDEs that are derived from the stochastic model will be the base of our control oriented model. In the following, we will discretize the PDEs~\eqref{eq:pdeOnMode}-\eqref{eq:pdeOffMode} and then show that the discretized model has special structure. Particularly, the structure elucidates how to model the aggregate under the effects of a arbitrary randomized policy. 

\begin{figure*}
	\centering
	\includegraphics[width=2\columnwidth]{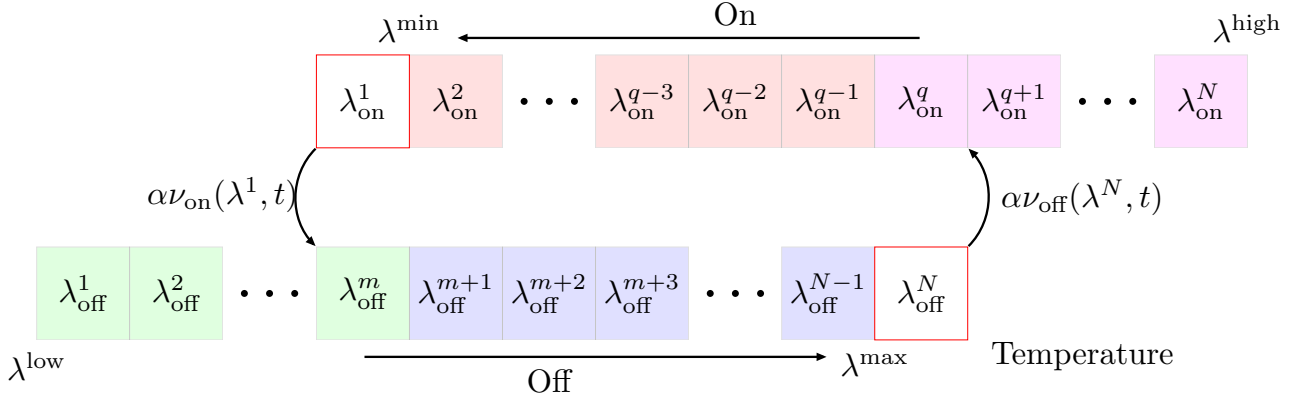}
	\caption{The control volumes (CVs). The colors correspond to the colors found in Figure~\ref{fig:sparPattern}. The values in each CV represent the nodal temperature for the CV. The arrows describe the sign of the convection of the TCL through the CVs.}
	\label{fig:cvLayout}
\end{figure*}

\section{Discretization} \label{sec:discretzation}
In order to be used, the coupled pdes~\eqref{eq:pdeOnMode}-\eqref{eq:pdeOffMode} need to be spatially and temporally discretized. We will use the finite volume method (FVM) to discretize~\eqref{eq:pdeOnMode} and~\eqref{eq:pdeOffMode}. In light of the discussion from the previous section, the goal will be to develop a control oriented model. That is, we aim to: (i) obtain a discretized model that agrees well with population quantities (avoids behavior as shown in Figure~\ref{fig:advEqPeriodic}) and (ii) discretize the model in a way that a control input for the BA can be identified. More on point (ii) will be discussed in Section~\ref{sec:indStruc}, however the discretization here plays a role. 


\subsection{Spatial discretization}
The layout of the control volumes (CV) is shown in Figure~\ref{fig:cvLayout}. The discretization is achieved by enumerating $N$, for both the on and off mode state, CV temperature values and their upper and lower boundaries:
\begin{align} \nonumber
	&\lambda_{\textOn} = (\lambda^{i}_{\textOn})_{i=1}^{N}, \ \lambda_{\textOn}^+ = \lambda_{\textOn} + \frac{\Delta\lambda}{2}, \ \lambda_{\textOn}^- = \lambda_{\textOn} - \frac{\Delta\lambda}{2}, \\ \nonumber
	&\lambda_{\textOff} = (\lambda^{i}_{\textOff})_{i=1}^{N}, \ \lambda_{\textOff}^+ = \lambda_{\textOff} + \frac{\Delta\lambda}{2}, \ \lambda_{\textOff}^- = \lambda_{\textOff} - \frac{\Delta\lambda}{2},
\end{align}
where $\Delta\lambda$ is the CV width.
All intermediate values of $\lambda_{\textOn}$ and $\lambda_{\textOff}$ are separated from each other by $\Delta\lambda$. The values in $\lambda^+_{\textOn}$ (respectively, $\lambda^+_{\textOff}$) are the right edges of the CVs and the values $\lambda^-_{\textOn}$ (respectively, $\lambda^-_{\textOff}$) are the left edges of the CVs, for example, $\lambda^{1,-}_{\textOff} = \lambda^{\text{low}}$. The quantities $\lambda^{\text{min}}$ and $\lambda^{\text{max}}$ specify the thermostat deadband, and are \emph{different} from the quantities $\lambda^{\text{high}}$ and $\lambda^{\text{low}}$ (see Figure~\ref{fig:cvLayout}).

We denote the $i^{th}$ CV as CV($i$) and further adopt the following notational simplifications,
\begin{align} \nonumber
	\mu_{\textOff}(\lambda^i,t) \triangleq \mu_{\textOff}(\lambda^i_{\textOff},t), \ \text{and} \ \mu_{\textOn}(\lambda^i,t) \triangleq \mu_{\textOn}(\lambda^i_{\textOn},t).
\end{align}
This simplification is extended to any situation that would otherwise require the double indication of the on or off state. Highlighted red in Figure~\ref{fig:cvLayout} are two additional control volumes. These control volumes are added to assist in enforcing boundary conditions that coincide with the thermostat control law~\eqref{eq:thermoContLaw}. Further discussion on this is given in section~\ref{sec:boundCond}.
\ifshowArxiv
\subsubsection{Internal CV's}
Consider the RHS of the pde~\eqref{eq:pdeOnMode} integrated over CV($i$):
\begin{align} \nonumber
	&\int_{\text{CV(i)}}\bigg(\frac{\sigma^2}{2}\frac{\partial^2}{\partial \lambda^2}\big(\mu_{\textOn}(\lambda,t)\big)-\frac{\partial}{\partial \lambda}\big(f_{\textOn}(\lambda,t)\mu_{\textOn}(\lambda,t)\big)\bigg)d\lambda \\
 \label{eq:arbOnIntCV}
&=\bigg(\frac{\sigma^2}{2}\frac{\partial}{\partial \lambda}\mu_{\textOn}(\lambda,t) -f_{\textOn}(\lambda,t)\mu_{\textOn}(\lambda,t) \bigg)\bigg\vert_{\lambda^{i,-}}^{\lambda^{i,+}},
\end{align}
where equality is by the divergence theorem~\cite{VersteegIntroductionBook:2007}.
Note, the points $\lambda^{i,-}$ and $\lambda^{i,+}$ are not control volume variables, but rather the boundaries of a single control volume. Hence, quantities in~\eqref{eq:arbOnIntCV} need to be approximated in terms of the nodal points of the neighboring control volumes. The approximations for the partial derivative are,
\begin{align} \label{eq:centDiff}
	\frac{\partial}{\partial \lambda}\mu_{\textOn}(\lambda^{i,+},t) \approx \frac{\mu_{\textOn}(\lambda^{i+1},t) - \mu_{\textOn}(\lambda^i,t)}{\Delta \lambda}, \\
	\frac{\partial}{\partial \lambda}\mu_{\textOn}(\lambda^{i,-},t) \approx \frac{\mu_{\textOn}(\lambda^{i},t) - \mu_{\textOn}(\lambda^{i-1},t)}{\Delta \lambda},
\end{align}
which correspond to a central difference approximation of the derivative. For the integrated convective term, we use the so-called upwind differencing scheme~\cite{VersteegIntroductionBook:2007}. This scheme elects the FVM equivalent of a forward or backward difference based on the sign of the convective velocity $f_{\text{on}}(\lambda,t)$. When the TCL is on $f_{\textOn}(\lambda,t) <0$ (i.e., the temperature decreases) the upwind differencing scheme prescribes:  
\begin{align}
	f_{\textOn}(\lambda^{i,-},t)\mu_{\textOn}(\lambda^{i,-},t) &= f_{\textOn}(\lambda^{i,-},t)\mu_{\textOn}(\lambda^{i},t), \ \text{and}\\
	f_{\textOn}(\lambda^{i,+},t)\mu_{\textOn}(\lambda^{i,+},t) &= f_{\textOn}(\lambda^{i,+},t)\mu_{\textOn}(\lambda^{i+1},t).
\end{align}
When the TCL is off $f_{\textOn}(\lambda,t) >0$ (i.e., the temperature increases) the upwind differencing scheme prescribes:  
\begin{align}
f_{\textOff}(\lambda^{i,-},t)\mu_{\textOff}(\lambda^{i,-},t) &= f_{\textOff}(\lambda^{i,-},t)\mu_{\textOff}(\lambda^{i-1},t),\\
f_{\textOff}(\lambda^{i,+},t)\mu_{\textOff}(\lambda^{i,+},t) &= f_{\textOff}(\lambda^{i,+},t)\mu_{\textOff}(\lambda^{i},t).
\end{align}
Now returning to the discretization of the PDE~\eqref{eq:pdeOnMode} over an arbitrary internal CV. We approximate the LHS of~\eqref{eq:pdeOnMode} integrated over the control volume as,
\begin{align} \nonumber
	\int_{\text{CV(i)}}\frac{\partial}{\partial t}\mu_{\textOn}(\lambda,t)d\lambda \approx \frac{d}{dt}\mu_{\textOn}(\lambda^i,t)\Delta\lambda = \frac{d}{dt}\nu_{\textOn}(\lambda^i,t),
\end{align}
where $\nu_{\textOn}(\lambda^i,t)\triangleq \mu_{\textOn}(\lambda^i,t)\Delta\lambda$.
We have used the ordinary differential as it will be the only differential to appear in the following. Now, denote the following
\begin{align}
	D \triangleq \frac{\sigma^2}{(\Delta\lambda)^2}, \quad \text{and} \quad F_{\textOn}^i(t) \triangleq \frac{f_{\textOn}(\lambda^{i},t)}{\Delta\lambda},
\end{align}
where the quantities $F_{\textOff}^i(t)$, $F_{\textOn}^{i,+}(t)/F_{\textOff}^{i,+}(t)$, and $F_{\textOn}^{i,-}(t)/F_{\textOff}^{i,-}(t)$ are defined analogously. Now inserting the central difference approximation and upwind difference approximation in~\eqref{eq:arbOnIntCV} and combining it with the approximation of the LHS of~\eqref{eq:pdeOnMode} we have,
\begin{align} \nonumber
\frac{d}{d t}\nu_{\textOn}(\lambda^i,t) &= \Big(F_{\textOn}^{i,-}(t) - D\Big)\nu_{\textOn}(\lambda^{i},t)+ \frac{D}{2}\nu_{\textOn}(\lambda^{i-1},t)  \\ \label{eq:stanOnCV}
&+\Big(\frac{D}{2}-F_{\textOn}^{i,+}(t)\Big)\nu_{\textOn}(\lambda^{i+1},t). 
\end{align}
The spatial discretization for the pde~\eqref{eq:pdeOffMode} is similar and yields,
\begin{align} \nonumber
\frac{d}{d t}\nu_{\textOff}(\lambda^i,t) &=  \frac{D}{2}\nu_{\textOff}(\lambda^{i+1},t)  -\Big(F_{\textOff}^{i,+}(t) + D\Big)\nu_{\textOff}(\lambda^{i},t)\\ \label{eq:stanOffCV}
&+\Big(\frac{D}{2}+F_{\textOff}^{i,-}(t)\Big)\nu_{\textOff}(\lambda^{i-1},t).
\end{align} 
\fi

\ifshowArxivAlt
\subsubsection{Internal CV's} We use a central difference to approximate the diffusion terms and the upwind scheme to approximate the convective terms. The textbook details of these schemes can be found in~\cite{VersteegIntroductionBook:2007} and the full derivation can be found in the arxiv version of this work~\cite{CoffmanControlArxiv:2020}. Now, denote the following
\begin{align}
D \triangleq \frac{\sigma^2}{(\Delta\lambda)^2}, \quad \text{and} \quad F_{\textOn}^i(t) \triangleq \frac{f_{\textOn}(\lambda^{i},t)}{\Delta\lambda},
\end{align}
where the quantities $F_{\textOff}^i(t)$, $F_{\textOn}^{i,+}(t)/F_{\textOff}^{i,+}(t)$, and $F_{\textOn}^{i,-}(t)/F_{\textOff}^{i,-}(t)$ are defined analogously. The resulting spatially discretized equations are
\begin{align} \nonumber
\frac{d}{d t}\nu_{\textOn}(\lambda^i,t) &= \Big(F_{\textOn}^{i,-}(t) - D\Big)\nu_{\textOn}(\lambda^{i},t)+ \frac{D}{2}\nu_{\textOn}(\lambda^{i-1},t)  \\ \label{eq:stanOnCV}
&+\Big(\frac{D}{2}-F_{\textOn}^{i,+}(t)\Big)\nu_{\textOn}(\lambda^{i+1},t), \\ 
\nonumber
\frac{d}{d t}\nu_{\textOff}(\lambda^i,t) &=  \frac{D}{2}\nu_{\textOff}(\lambda^{i+1},t)  -\Big(F_{\textOff}^{i,+}(t) + D\Big)\nu_{\textOff}(\lambda^{i},t)\\ \label{eq:stanOffCV}
&+\Big(\frac{D}{2}+F_{\textOff}^{i,-}(t)\Big)\nu_{\textOff}(\lambda^{i-1},t),
\end{align} 
where  $\nu_{\textOn}(\lambda^i,t)\triangleq \mu_{\textOn}(\lambda^i,t)\Delta\lambda$, and $\nu_{\textOff}(\lambda^i,t)\triangleq \mu_{\textOff}(\lambda^i,t)\Delta\lambda$.
\fi

\subsubsection{Boundary CV's} \label{sec:boundCond}
The boundary CV's are the CVs associated with the nodal values: $\lambda_{\textOn}^1$, $\lambda_{\textOn}^q$, $\lambda_{\textOn}^N$, $\lambda_{\textOff}^1$, $\lambda_{\textOff}^m$, and $\lambda_{\textOff}^N$. The superscript, for example the integer $q$ in $\lambda_{\textOn}^q$ represents the CV index. All boundary CVs can be seen in Figure~\ref{fig:cvLayout}. Discretization of the boundary CV's requires care for atleast two reasons. First, this is typically where one introduces the BCs of the pde into the numerical approximation. Secondly, on finite domains the endpoints present challenges as, for example, there is no variable $\mu_{\textOff}(\lambda^{N+1},t)$ for computation of the first partial derivative values for node $\lambda^{N+1}_\textOff$.

The BC's for the coupled PDEs~\eqref{eq:pdeOnMode}-\eqref{eq:pdeOffMode} are~\cite{MalhameElectricTAC:1985}:
\begin{align} \nonumber
	&\textbf{Absorbing Boundaries:} \\ \label{eq:absorbBoundary}
	&\qquad \qquad \qquad \mu_{\textOn}(\lambda^{\text{min}},t) = \mu_{\textOff}(\lambda^{\text{max}},t) = 0. \\ \nonumber
	&\textbf{Conditions at Infinity:} \\ \label{eq:contInfBC}
	&\qquad \qquad \qquad \mu_{\textOn}(+\infty,t) = \mu_{\textOff}(-\infty,t) = 0. \\ \nonumber
	&\textbf{Conservation of Probability:} \\ \label{eq:consProbOne}
	&\frac{\partial}{\partial\lambda}\bigg[\mu_{\textOn}(\lambda^{q,-},t)-\mu_{\textOn}(\lambda^{q-1,+},t)  -\mu_{\textOff}(\lambda^{N-1,+},t)\bigg]= 0. \\ \label{eq:consProbTwo}
	&\frac{\partial}{\partial\lambda}\bigg[\mu_{\textOff}(\lambda^{m,+},t)-\mu_{\textOn}(\lambda^{2,-},t) -\mu_{\textOff}(\lambda^{m+1,-},t)\bigg]= 0. \\
	&\nonumber \textbf{Continuity:} \\ \label{eq:contBC_1}
	&\qquad \qquad \qquad \mu_{\textOn}(\lambda^{q,-},t) = \mu_{\textOn}(\lambda^{q-1,+},t). \\  \label{eq:contBC_2}
	&\qquad \qquad \qquad \mu_{\textOff}(\lambda^{m,+},t) = \mu_{\textOff}(\lambda^{m+1,-},t). 
\end{align}
As we will see, implementation of some of the above conditions will require a bit of care. However, some are quite trivial to enforce. For example, by default, the continuity conditions~\eqref{eq:contBC_1} and~\eqref{eq:contBC_2} are satisfied due to our choice of CV structure, since, for example, for any $i$ we have $\lambda_{\textOff}^{i,-}=\lambda_{\textOff}^{i-1,+}$ and $\lambda_{\textOff}^{i,+}=\lambda_{\textOff}^{i+1,-}$. 
\ifx 0
Further, the choice of grid also reduce the conservation of probability to the:
\begin{align} \label{eq:modConsProb}
	\frac{\partial}{\partial \lambda}\mu_{\textOff}(\lambda^{N-1,+},t) = \frac{\partial}{\partial \lambda}\mu_{\textOff}(\lambda^{2,-},t) = 0,
\end{align}
which with the BC~\eqref{eq:absorbBoundary} specifies a periodic boundary condition. 
We can reduce~\eqref{eq:consProbOne} and~\eqref{eq:consProbOne} to~\eqref{eq:modConsProb} since the FVM, by design, conserves.Although, we return to this condition after dealing with the simpler conditions first.   
\fi

Now focusing on the conditions at infinity BC~\eqref{eq:contInfBC}, we enforce instead the following conditions:
\begin{align}
	\frac{\partial}{\partial \lambda}\mu_{\textOff}(\lambda^{1,-},t) = 0, \quad \text{and} \quad
	\frac{\partial}{\partial \lambda}\mu_{\textOn}(\lambda^{N,+},t) = 0.
\end{align}
The reason for this is because our computational domain cannot extend to infinity, where the BC~\eqref{eq:contInfBC} is required to hold. Practically, the temperature values $\lambda_{\textOff}^1$ and $\lambda_{\textOn}^N$ are quite far away from the deadband and so the density here will be near zero anyways.  

Now, consider the spatial discretization of the CVs associated with the BC at infinity. First considering the CV associated with the temperature $\lambda^1_{\textOff}$, we have that the differential equation is
\begin{align} \nonumber
\frac{d}{d t}\nu_{\textOff}(\lambda^1,t) &= \Big(-F_{\textOff}^{1,+}(t) - \frac{D}{2}\Big)\nu_{\textOff}(\lambda^{1},t)  \\ \label{eq:BC_lambdaOff1}
&+\Big(\frac{D}{2}+F_{\textOff}^{2,-}(t)\Big)\nu_{\textOff}(\lambda^{2},t).
\end{align}
Considering the CV associated with the temperature $\lambda^N_{\text{on}}$, we have
\begin{align} \nonumber
\frac{d}{d t}\nu_{\textOn}(\lambda^N,t) &= \Big(F_{\textOn}^{N,+}(t) - \frac{D}{2}\Big)\nu_{\textOn}(\lambda^{N},t)  \\ \label{eq:BC_lambdaOnN}
&+\Big(\frac{D}{2}-F_{\textOn}^{N,+}(t)\Big)\nu_{\textOn}(\lambda^{N-1},t).
\end{align}
In the above we make the assumption that $\nu_{\textOff}(\lambda^{1,-} - \Delta\lambda,t) = 0$ and $\nu_{\textOn}(\lambda^{N,+} + \Delta\lambda,t) = 0$. 

Now focus on the absorbing boundary~\eqref{eq:absorbBoundary} and conservation of probability~\eqref{eq:consProbOne}-\eqref{eq:consProbTwo} boundary conditions. These BCs have the following meaning. The condition~\eqref{eq:absorbBoundary} clamps the density at the end of the deadband to zero. BC~\eqref{eq:consProbOne} reads: the net-flux across the temperature value $\lambda^q_{\textOn}$ is equal to the flux of density going from off to on. In order to enforce both~\eqref{eq:consProbOne} and \eqref{eq:consProbTwo} we will model the flux due to TCLs switching as sources/sinks. Before doing this, we mention some conceptual issues with enforcing the BC~\eqref{eq:absorbBoundary}.

Problematically, a TCL's state trajectory will never satisfy the BC~\eqref{eq:absorbBoundary} since to switch its mode state the TCLs temperature sensor will have to register a value outside the deadband. That is, it is possible to enforce the BC~\eqref{eq:absorbBoundary}, however the developed model would have a discrepancy with the behavior of a TCL. To combat this, we introduce two additional CV's associated with the temperatures $\lambda^1_\textOn$ and $\lambda^N_\textOff$, which are the ones outlined in red in Figure~\ref{fig:cvLayout}. We then transfer the BC~\eqref{eq:absorbBoundary} to one on the added CVs, where the transferred BC is now
\begin{align} \label{eq:modfAbsCond}
	\mu_{\textOn}(\lambda^{1,-},t) = \mu_{\textOff}(\lambda^{N,+},t) &= 0.
\end{align}
As mentioned, to enforce the conservation of probability BC we use a source/sink type argument, which we also enforce on the added CVs. To see what we mean by source/sink argument, consider the following: some rate of TCLs are transferred out of the CV $\lambda^N_\textOff$ and into the CV $\lambda^q_\textOn$ (as depicted in Figure~\ref{fig:cvLayout}) due to thermostatic control. Since during operation, any TCL within the CV $\lambda^N_\textOff$ would immediately switch on, we model the sink as simply $-\nu_{\textOff}(\lambda^{N},t)$. The rate of the sink is then given as $-\gamma\nu_{\textOff}(\lambda^{N},t)$, where $\gamma>0$ is a modeling choice and a constant of appropriate units that describes the discharge rate. We shortly given insight on how to elect a value for $\gamma$. Now discretizing the CV corresponding to the nodal value $\lambda^N_{\textOff}$ subject to the BC~\eqref{eq:modfAbsCond} and the sink $-\nu_{\textOff}(\lambda^{N},t)$ we obtain,
\begin{align} \label{eq:BC_lambdaOffN}
	\frac{d}{dt}\nu_{\textOff}(\lambda^{N},t) &= \Big(\frac{D}{2}+F_{\textOff}^{N,-}(t)\Big)\nu_{\textOff}(\lambda^{N-1},t) -\alpha\nu_{\textOff}(\lambda^{N},t), 
\end{align}
where $\alpha = \big(\gamma+ D\big)$. In obtaining the above, we have made the reasonable assumption that $\nu_{\textOff}(\lambda^{N,+}+\Delta\lambda,t) = 0$.  The quantity $\alpha\nu_{\textOff}(\lambda^{N},t)$ represents the rate of change of density from the CV $\lambda^N_{\textOff}$ to the CV $\lambda^q_{\textOn}$, as depicted in Figure~\ref{fig:cvLayout}. Consequently, to conserve probability, we must add this quantity as a source to the ode for the CV $\lambda^q_{\textOn}$, i.e., 
\begin{align}\label{eq:boundOff2On}
	\frac{d}{dt}\nu_{\textOn}(\lambda^{q},t) &= \dots + \alpha\nu_{\textOff}(\lambda^{N},t). 
\end{align}
The dots in equation~\eqref{eq:boundOff2On} represent the portion of the dynamics for the standard internal CV (i.e., the RHS of~\eqref{eq:stanOnCV}) for the temperature node $\lambda^q_{\textOn}$. A similar argument is used for the BC~\eqref{eq:consProbTwo} with the CV's $\lambda^{1}_{\textOn}$ and $\lambda^{m}_{\textOff}$, and the corresponding differential equations are,
\begin{align} \label{eq:BC_lambdaOn1}
\frac{d}{dt}\nu_{\textOn}(\lambda^{1},t) &= \Big(\frac{D}{2}-F_{\textOn}^{1,+}(t)\Big)\nu_{\textOn}(\lambda^{2},t) - \alpha\nu_{\textOn}(\lambda^{1},t), \\ \label{eq:boundOn2Off}
\frac{d}{dt}\nu_{\textOff}(\lambda^{m},t) &= \dots + \alpha\nu_{\textOn}(\lambda^{1},t). 
\end{align}
Practically, once the differential equations are discretized in time with timestep $\Delta t$ one will then elect $\gamma$ so that $\alpha = (\Delta t)^{-1}$. With this choice, the discretized equations have the interpretations that all mass starting in state $\nu_{\textOff}(\lambda^{N},\cdot)$ at time $t$ is transferred out by time $t+\Delta t$ into the state $\nu_{\textOn}(\lambda^{q},\cdot)$. 

\subsubsection{Overall system}
\begin{figure}
	\centering
	\includegraphics[width=0.75\columnwidth]{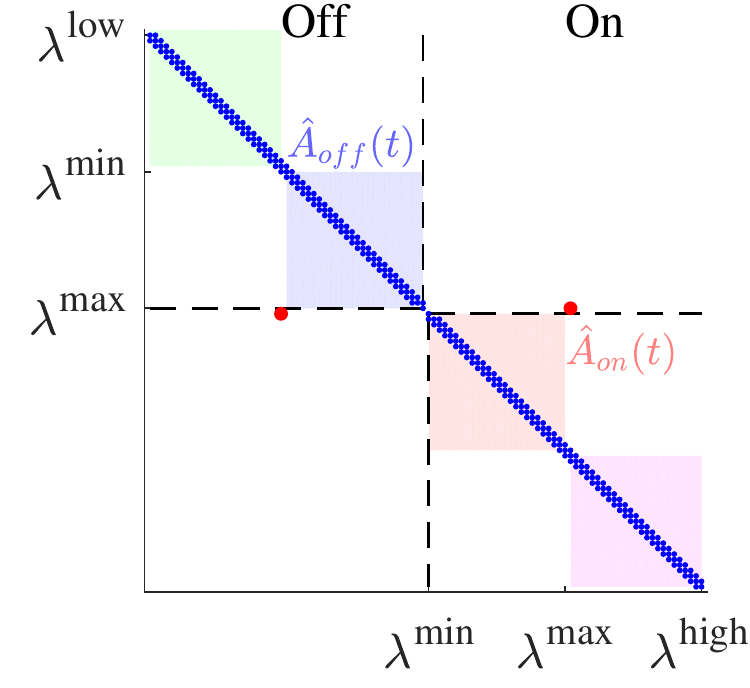}
	\caption{Sparsity pattern of the matrix $A(t)$ for $N=51$ CVs for both the on and off state. The colors correspond to the colors found in Figure~\ref{fig:cvLayout}.}
	\label{fig:sparPattern}
\end{figure}
Denoting the state of the overall system at time $t$ as the row vector, $\nu(t) = [\nu_{\textOff}(t),\nu_{\textOn}(t)]$ with
\begin{align}
	\nu_{\textOff}(t) &= [\nu_{\textOff}(\lambda^1,t), \dots,\nu_{\textOff}(\lambda^N,t)], \\ 
	\nu_{\textOn}(t) &=[\nu_{\textOn}(\lambda^1,t), \dots,\nu_{\textOn}(\lambda^N,t)],
\end{align}
and combing the odes: \eqref{eq:stanOnCV} and~\eqref{eq:stanOffCV} for all of the internal CVs and \eqref{eq:BC_lambdaOff1}, \eqref{eq:BC_lambdaOnN}, \eqref{eq:BC_lambdaOffN}, \eqref{eq:boundOff2On}, \eqref{eq:BC_lambdaOn1}, \eqref{eq:boundOn2Off} for the BC CVs. We obtain the linear time varying system,
\begin{align} \label{eq:matContDynNotTran}
	\frac{d}{dt}\nu^T(t) = \mathcal{A}(t)\nu^T(t).
\end{align}
The matrix $\mathcal{A}(t)$ contains all of the coefficients from the individual ode's developed so far from spatial discretization. In the following, it will be convenient to view the dynamics~\eqref{eq:matContDynNotTran} in their transposed form
\begin{align} \label{eq:dynContMC}
\frac{d}{dt}\nu(t) = \nu(t)A(t),
\end{align}
with $A(t) = \mathcal{A}^T(t)$. We have included the sparsity pattern of $A(t)$ in Figure~\ref{fig:sparPattern}. The matrix $A(t)$ also satisfies the properties of a transition rate matrix, described in the following lemma.
\begin{lem} \label{lem:rateMat}
	For all $t$, the matrix $A(t)$ is a transition rate matrix, that is, it satisfies for all $t$,
	\begin{align} \nonumber
	\text{(i):} &\quad A(t)\mathbb{1} = 0. \\ \nonumber
	\text{(ii):} &\quad \forall \ i, \ A_{i,i}(t) \leq 0, \ \text{and} \ \forall \ j\neq i \  A_{i,j}(t) \geq 0. 
	\end{align}
\end{lem}
\ifshowArxivAlt
\begin{proof}
See extended arxiv version~\cite{CoffmanControlArxiv:2020}.
\end{proof}
\fi
\ifshowArxiv
\begin{proof}
	The FVM method is well known to conserve mass, so that property (i) is readily satisfied. Property (ii) can be inferred from the each individual CV equation.
\end{proof}
\fi

\subsection{Temporal discretization}
To temporally integrate the dynamics~\eqref{eq:dynContMC} we use a first order Euler approximation with time step $\Delta t>0$. Making the identifications $\nu_{k+1}\triangleq\nu(t_{k+1})$, $\nu_{k}\triangleq\nu(t_k)$, and $A_k \triangleq A(t_k)$ we have
\begin{align} \label{eq:discDynEsem}
	\nu_{k+1} = \nu_kP_k, \ \text{with} \ P_k = I + \Delta tA_k.
\end{align}
In the continuous time setting elements of the vector $\nu(t)$ were referred to as, for example, $\nu_{\textOn}(\lambda^i,t)$. The counterpart to this, in the discrete time setting, is referring to elements of $\nu_k$ as, for example, $\nu_{\textOn}[\lambda^i,k]$. 

\section{Identifying structure and the control input} \label{sec:indStruc}
We started with the PDEs~\eqref{eq:pdeOnMode}-\eqref{eq:pdeOffMode}, and in the previous section completely discretized them. Recall, that in the original work~\cite{MalhameElectricTAC:1985} the PDEs were developed under the assumption that the mode state $m(t)$ evolved according to~\eqref{eq:thermoContLaw}. Hence, from the viewpoint of control, we still need to identify the control input since the final discretized model~\eqref{eq:discDynEsem} has no control input. The goal of this section is to identify any structure that may be present in the matrix $P_k$ appearing in~\eqref{eq:discDynEsem} and to then exploit it for purposes of introducing a control input. Key to doing this is the result that $P_k$ is a transition matrix.
\begin{lem}\label{lem:CFL}
	Denote the $i^{th}$ diagonal element of the matrix $A_k$ as $[A_k]_{i,i}$. The matrix $P_k$ is a transition matrix if,
	\begin{align} \nonumber
	\forall \ i, \ \text{and} \ \forall \ k, \quad 0 < \Delta t \leq \frac{1}{\left|[A_k]_{i,i}\right|}.
	\end{align}
\end{lem}
\begin{proof}
	From Lemma~\ref{lem:rateMat} we have that $P_k\mathbb{1} = I\mathbb{1} + \Delta tA_k\mathbb{1} = \mathbb{1}$ since $A_k\mathbb{1} = 0$. Also from Lemma~\ref{lem:rateMat}, every element of $A_k$ is non-negative, save for the diagonal elements. Under the hypothesis on $A_k$, then every diagonal element of $I + \Delta tA_k$ will be in $[0,1]$. 
\end{proof}
Now, when the conditions of Lemma~\ref{lem:CFL} are met $P_k$ is a transition matrix and hence each $\nu_k$ can be viewed as a marginal distribution if $\nu_0\mathbb{1} = 1$ and $\nu[\cdot,0] > 0$. The structure of this marginal is given from~\eqref{eq:probOnState} for the on state (a similar interpretation holds for the off state) as,
\begin{align} 
\nu_{\textOn}[\lambda^i,k] 
&=\Prob\left(x(t_k) \in \text{CV}(i), \ m(t_k) = \text{on} \right),
\end{align} 
where $x(t_k)$ is the temperature. Now denote, $x_k \triangleq x(t_k)$, $m_k \triangleq m(t_k)$, and
\begin{align} \label{eq:binnedState}
	I_k \triangleq \sum_{i=1}^{N}i\mathbf{I}\Big(x_k \in \text{CV}(i)\Big),
\end{align}
where $\mathbf{I}(\cdot)$ is the indicator function. The quantity $I_k$ can be thought of as a ``binned'' state that indicates the CV index. Using $I_k$ we rewrite $\nu_{\textOn}[\lambda^i,k]$ and $\nu_{\textOff}[\lambda^i,k]$ as,
\begin{align} \label{eq:margNu}
	\nu_{\textOn}[\lambda^i,k] &= \Prob\left(I_k = i, \ m_k = \text{on} \right), \ \text{and} \\
	\nu_{\textOff}[\lambda^i,k] &= \Prob\left(I_k = i, \ m_k = \text{off} \right).
\end{align}

\subsection{Conditional independence of $P_k$}
From~\eqref{eq:margNu}, the matrix $P_k$ (with the conditions of Lemma~\ref{lem:CFL} satisfied) is the transition matrix for the joint process $(I_k,m_k)$. In the following, we refer to the values of $I_k$ with $i$ and $j$ and the values of $m_k$ with $u$ and $v$. We introduce the following notation to refer to the elements of the transition matrix $P_k$: $P_w((i,u),(j,v))\triangleq$
\begin{align} \nonumber
\Prob\Big(I_{k+1} = j, \ m_{k+1} = v \ \vert \ I_k = i, \  m_k = u, \ \theta^a_k = w\Big).
\end{align}
Recall, the matrix $P_k$ is derived for the nominal thermostat policy. We will now show that the matrix $P_k$ can be written as the product of two matrices. One depending on the nominal thermostat policy and one depending on weather and TCL system dynamics. That is, to show that each element of the matrix $P_k$ can be written as,
\begin{align} \label{eq:condIndFact}
	P_w((i,u),(j,v)) = \phi_u(v \ \vert \ j)G_w((i,u),j)
\end{align}
where: $G_w((i,u),j)\triangleq$
\begin{align} 
	&\Prob\left(I_{k+1} = j \ \vert \ I_k = i, \ m_k = u, \ \theta^a_k = w  \right), \ \text{and}\\ 
	&\phi_u(v \ \vert \ j) \triangleq\Prob\left(m_{k+1} = v \ \vert \ I_{k+1} = j, \ m_{k} = u \right).
\end{align}
The quantity $\phi_u(v \ \vert \ j)$ is the factor that depends on the nominal thermostat policy. As such in the following we denote $\phi_u(v \ \vert \ j)$ as a policy. The vectorized form of the policies are, 
\begin{align}
\phi_{\textOff} \triangleq
\phi_{\textOff}(\text{on} \ \vert \ \cdot), \quad  \text{and} \quad \phi_{\textOn} \triangleq
\phi_{\textOn}(\text{off} \ \vert \ \cdot),
\end{align}
where $\phi_{\textOff},\ \phi_{\textOn}\in \mathbb{R}^N$. The factorization~\eqref{eq:condIndFact} is represented in matrix form as,
\begin{align} \label{eq:condIndFactMat}
	P_k = \Phi G_k,
\end{align}
where $\Phi \in \mathbb{R}^{2N\times4N}$ and $G_k \in \mathbb{R}^{4N\times2N}$. The subscript $k$ on $G_k$ is to denote its dependence on the time varying ambient temperature $\theta^a_k$. The factorization~\eqref{eq:condIndFactMat} is paramount as it tells us how the nominal thermostat policy and weather independently contribute to the overall dynamics. Inversely, it then informs us how to define the matrix $P_k$ for a different (possibly randomized) policy.

We show the factorization~\eqref{eq:condIndFactMat} through construction, i.e., we find a matrix $\Phi$ and $G_k$ that simultaneously satisfy~\eqref{eq:condIndFactMat} and~\eqref{eq:condIndFact}. We start this construction through the sparsity structure shown for $A(t)$ in Figure~\ref{fig:sparPattern}. Based on shaded regions of the matrix $A(t)$ shown in Figure~\ref{fig:sparPattern}, we define the following:
\begin{align}
	P^{\textOn}_k = I +\Delta tA^{\textOn}_k, \ \text{and} \  P^{\textOff}_k = I +\Delta tA^{\textOff}_k, \\
	\hat{P}^{\textOn}_k = I +\Delta t\hat{A}^{\textOn}_k, \ \text{and} \  \hat{P}^{\textOff}_k = I +\Delta t\hat{A}^{\textOff}_k.
\end{align}
More precisely $A^{\textOn}_k$ (respectively, $A^{\textOff}_k$) is the matrix containing  the coefficients of the spatially discretized PDE~\eqref{eq:pdeOnMode} (respectively, PDE~\eqref{eq:pdeOffMode}) evaluated at time $t_k$. That is, $A^{\textOn}_k$ is the matrix that corresponds to the bottom-right quadrant encompassed by the dashed black line in Figure~\ref{fig:sparPattern}. The matrix $\hat{A}^{\textOn}_k$ (respectively, $\hat{A}^{\textOff}_k$) holds the same interpretation as $A^{\textOn}_k$ (respectively, $A^{\textOff}_k$) except restricted to the control volumes between $[\lambda^{\text{min}},\lambda^{\text{max}}]$. We additionally define the following matrices,
\begin{align}
	S^{\textOn}_k = \begin{bmatrix}
	\mathbf{0} & \mathbf{0} \\
	\hat{P}_k^{\textOn} & \mathbf{0} 
	\end{bmatrix}, \ \text{and} \ 
	S^{\textOff}_k = \begin{bmatrix}
	\mathbf{0} & \hat{P}_k^{\textOff} \\
	\mathbf{0} & \mathbf{0} 
	\end{bmatrix}
\end{align}
which will be used to show the factorization~\eqref{eq:condIndFact} in the following lemma. The size of the zero matrices in both $S^{\textOn}_k$ and $S^{\textOff}_k$ are such that the size of the matrices $S^{\textOn}_k$ and $S^{\textOff}_k$ is the same as $P^{\textOn}_k$ and $P^{\textOff}_k$.
\begin{lem} \label{lem:condFact}
	Let $e^{\ell}$ be the $\ell^{th}$ canonical basis vector in $\mathbb{R}^N$. Let $\phi_{\textOff} = e^1$ and $\phi_{\textOn} = e^N$ then denote $\Phi_{\textOn} = \text{diag}(\phi_{\textOn})$ and $\Phi_{\textOff} = \text{diag}(\phi_{\textOff})$. Now, let $\Phi$ and $G_k$ be given as,
	\begin{align} \label{eq:polMatInLem}
		\Phi &= \begin{bmatrix}
			I - \Phi_{\textOff} & \Phi_{\textOff} & 0 & 0 \\
			0 & 0 & \Phi_{\textOn} & I - \Phi_{\textOn}
		\end{bmatrix} \\
		G_k &= \begin{bmatrix}
		0 & S_k^{\textOn} & 0 & P_k^{\textOff} \\
		P_k^{\textOn} & 0 & S_k^{\textOff} & 0
		\end{bmatrix}^T.
	\end{align}
	If $\alpha = (\Delta t)^{-1}$ (appearing in~\eqref{eq:boundOff2On} and~\eqref{eq:boundOn2Off}),
	then $P_k = \Phi G_k.$
\end{lem}
\ifshowArxivAlt
\begin{proof}
	See extended arxiv version~\cite{CoffmanControlArxiv:2020}.
\end{proof}	
\fi
	\ifshowArxiv
	\begin{proof}
	If $\alpha = (\Delta t)^{-1}$, the diagonal elements of $A_k$ with $\alpha$ in them will go to zero and the non diagonal elements will go to 1. These non-diagonal elements with value $1$ are the red dots in Figure~\ref{fig:sparPattern} and encapsulate the thermostat control law. Thus the construction of $\Phi$ with the canonical basis vectors. Now, multiplying out the matrix we have,
	\begin{align}
		\Phi G_k = \begin{bmatrix}
		\big(I-\Phi_{\textOff}\big)P_k^{\textOn} & \Phi_{\textOff}S_k^{\textOn} \\
		\Phi_{\textOn}S_k^{\textOn} & \big(I-\Phi_{\textOn}\big)P_k^{\textOn}
		\end{bmatrix}
	\end{align}
	where $\big(I-\Phi_{\textOff}\big)P_k^{\textOn}$ (respectively, $\big(I-\Phi_{\textOn}\big)P_k^{\textOff}$) is the matrix $P_k^{\textOn}$ (respectively, $P_k^{\textOff}$) but with the last (respectively, first) row zeroed out. The exact opposite statement is true for $\Phi_{\textOff}P_k^{\textOn}$ and $\Phi_{\textOn}P_k^{\textOff}$. Hence, by definition of the matrices in $G_k$ we have $P_k = \Phi G_k$ where each non-zero element holds the interpretation~\eqref{eq:condIndFact}.
\end{proof}
	\fi

The conditional independence factorization has been a useful assumption in the design of algorithms in~\cite{busmey:CDC:2016}. In the present it is a byproduct of our spatial and temporal discretization of the PDE's~\eqref{eq:pdeOnMode}-\eqref{eq:pdeOffMode}. There are at least two important consequences of the factorization result from Lemma~\ref{lem:condFact}. The first one is described in the following corollary.
\begin{cor} \label{cor:detPol}
	For the nominal thermostat policy~\eqref{eq:thermoContLaw}, our spatial and temporal discretization scheme induces the degenerate (deterministic) stationary policy:
	\begin{align} \nonumber
		&\Prob(m_k = \text{on} \ \vert \ I_k = N, \ m_{k-1} = \text{off}) = 1, \\ \nonumber
		&\Prob(m_k = \text{off} \ \vert \ I_k = 1, \ m_{k-1} = \text{on}) \ = 1,
	\end{align}
	and zero otherwise.
\end{cor}
\begin{proof}
	Identifying the non zero elements of the policies $\phi_{\textOn}$ and $\phi_{\textOff}$ in Lemma~\ref{lem:condFact} with the respective state values gives the desired result.
\end{proof}
Hence, the policy induced by the nominal thermostat policy~\eqref{eq:thermoContLaw} and described in Corollary~\ref{cor:detPol} is exactly the  nominal thermostat policy. This recovery of the original control law gives confidence in the underlying spatial and temporal discretization schemes. The second important consequence of Lemma~\ref{lem:condFact} is that it informs us how to define the dynamics of the marginals~\eqref{eq:margNu} under a different policy than the nominal thermostat policy. 

\subsection{Introducing control + aggregate model} \label{sec:contAggMod}
In light of Lemma~\ref{lem:condFact}, we can now introduce an arbitrary randomized policy in place of the degenerate nominal thermostat policies described in Corollary~\ref{cor:detPol}. From the viewpoint of the BA this  randomized policy \emph{is} the control input. To distinguish from $\phi_{\textOff}$ and $\phi_{\textOn}$ in the prior section we denote the newly introduced policies with the superscript `BA' and describe them as `BA control policies.' For example, electing policies $\phi^{\textBA}_{\textOn}$ and $\phi^{\textBA}_{\textOff}$ as,
\begin{align} \label{eq:randPolOff2On}
&\phi^{\textBA}_{\textOff}(\text{on} \ \vert \ j) = \begin{cases}
	\kappa^{\textOn}_j, & (m+1) \leq j \leq (N-1). \\
	1, & j = N. \\
	0, & \text{o.w.}
	\end{cases} \\ \label{eq:randPolOn2Off}
	&\phi^{\textBA}_{\textOn}(\text{off} \ \vert \ j) = \begin{cases}
	\kappa^{\textOff}_j, & 2 \leq j \leq (q-1). \\
	1, & j = 1. \\
	0, & \text{o.w.}
	\end{cases}
\end{align}
with $\phi^{\textBA}_{\textOff}(\text{off} \ \vert \ \cdot) = 1 - 	\phi^{\textBA}_{\textOff}(\text{on} \ \vert \ \cdot)$ and $\phi^{\textBA}_{\textOn}(\text{on} \ \vert \ \cdot) = 1 - 	\phi^{\textBA}_{\textOn}(\text{off} \ \vert \ \cdot)$ and $\kappa^{\textOn}_j,\kappa^{\textOff}_j \in [0,1]$ for all $j$ will preserve the factorization interpretation found in Lemma~\ref{lem:condFact}. The policies could also be time varying, for example: $\kappa^{\textOff}_j[k]$ and $\kappa^{\textOn}_j[k]$. The dependence of the policies on time is denoted as $\phi^{\textBA}_{\textOff}[k]$ and $\phi^{\textBA}_{\textOn}[k]$. 

We have required $\phi^{\textBA}_{\textOff}(\text{on} \ \vert \ j) = 0$ for $1\leq j \leq m$ since the temperatures corresponding to these indices are below the permitted deadband temperature, $\lambda^{\text{min}}$. Hence, turning on at these temperature does not make physical sense. The arguments for the zero elements in $\phi^{\textBA}_{\textOn}$ are symmetric.

\begin{remark} \label{rem:impRandPol}
From the individual TCLs perspective, implementation of BA control policies of the form~\eqref{eq:randPolOff2On}-\eqref{eq:randPolOn2Off} is straightforward: (i) the TCL measures its current state, (ii) the TCL ``bins'' this state value according to~\eqref{eq:binnedState} and (iii) the TCL flips a coin to decide its next on/off state according to the probabilities given in~\eqref{eq:randPolOff2On}-\eqref{eq:randPolOn2Off}. Note that the randomized policies are wrapped inside of the nominal thermostat policy~\eqref{eq:thermoContLaw}, so that both the BA control policy and nominal thermostat policies are equivalent in enforcing the temperature constraint.
\end{remark}

In the following, we denote $\Phi^{\textBA}_k$ as the matrix with structure~\eqref{eq:polMatInLem} but containing any time varying BA control policies $\phi^{\textBA}_{\textOff}[k]$ and $\phi^{\textBA}_{\textOn}[k]$ that satisfy the requirements specified in~\eqref{eq:randPolOff2On} and~\eqref{eq:randPolOn2Off}, respectively. With this, the \emph{control oriented aggregate model} is the following discrete time system
\begin{align} \label{eq:aggModel}
\nu_{k+1} = \nu_k\Phi^{\textBA}_k G_k, \quad \gamma_k = \nu_kC_{\textOn},
\end{align}
where $C_{\textOn} = [\mathbf{0}^T, P_{\textAgg}\mathbb{1}^T]^T$ with $P_{\textAgg} \triangleq P\numTCLs$. The control input for this model is the policy $\Phi^{\textBA}_k$, which can be implemented as a control input at each TCL (see previously Remark~\ref{rem:impRandPol}). The nominal consumption for the ensemble expressed in terms of the nominal consumption of the individual TCL~\eqref{eq:invTCLbase} is, 
\begin{align}
\bar{P}_k \triangleq \numTCLs\bar{P}^{\text{ind}}_k.
\end{align}
The nominal consumption is time varying due to its dependence on the time varying ambient temperature. This quantity, modulo a constant, represents the fraction of TCLs that are on in nominal operation.
\section{Numerical Examples} \label{sec:aggModel}
We now conduct numerical experiments to show: (i) how the PDE's~\eqref{eq:pdeOnMode}-\eqref{eq:pdeOffMode} can be used to model an ensemble of TCLs and (ii) how the framework can be used to design BA control policies so that the ensemble of TCLs track a power reference signal. Each TCL is indexed by $\ell$ and the total number of TCLs is denoted \numTCLs. For example, $m_k^\ell$ and $I_k^\ell$ are the mode and binned temperature of the $\ell^{th}$ TCL at time $k$.
\subsection{Evaluating the aggregate model}
Two empirical ensemble quantities of interest are:
\begin{align} \nonumber
	Y_k \triangleq P\sum_{\ell=1}^{\numTCLs}m_k^\ell, \ \text{and} \ H_k[i,u] \triangleq \sum_{\ell=1}^{\numTCLs} \mathbf{I}\big(I^{\ell}_k = i, m_k^\ell = u\big), 
\end{align}
which are the total power consumption and histogram of the ensemble, respectively. They are empirical counterparts to the analytical quantities described through the aggregate model~\eqref{eq:aggModel}.

We now compare the empirical and analytical aggregate quantities in simulation. The results are shown in Figure~\ref{fig:histSimCompare} and~\ref{fig:powerSimCompare} for $\numTCLs = 50000$. The mode state of each TCL evolves according to a BA control policy that satisfies the structural requirements in~\eqref{eq:randPolOff2On} and~\eqref{eq:randPolOn2Off} and is relatively similar to the nominal thermostat policy (Corollary~\ref{cor:detPol}). The temperature evolution evolves according to a simulated version of~\eqref{eq:stoModelTCL}. We see the state $\nu_k$ matches well the histogram $H_k$ of the ensemble (Figure~\ref{fig:histSimCompare}) and the output $\gamma_k$ matches well the ensembles power consumption $Y_k$ (Figure~\ref{fig:powerSimCompare}).   
\subsection{Controlling the Ensemble}  
Due to space limitations, a full description of the control algorithm is not possible. However, as a preview we present simulation results from the algorithm in Figure~\ref{fig:powerTrackCent}. The reference signal $r_k$ shown in Figure~\ref{fig:powerTrackCent} is an arbitrarily generated sum of sinusoids added to the nominal power, $\bar{P}_k$. The ambient air temperature is time varying and is obtained from  weatherunderground.com for a typical summer day in Gainesville, Fl.

The control algorithm amounts to solving a convex optimization problem, and its facilitation is in large part due to the identified structure. Essentially, the optimization problem utilizes the model~\eqref{eq:aggModel} to obtain a string of optimal randomized BA control policies $\Phi^{\textBA,*}_k$. The BA can then send these policies to each TCL, where implementation is as described in~\ref{rem:impRandPol}. Each TCL using the designed BA control policies has the effect of the ensemble tracking $r_k$, as shown in Figure~\ref{fig:powerTrackCent}.

\begin{figure}
	\centering
	\includegraphics[width = 1\columnwidth]{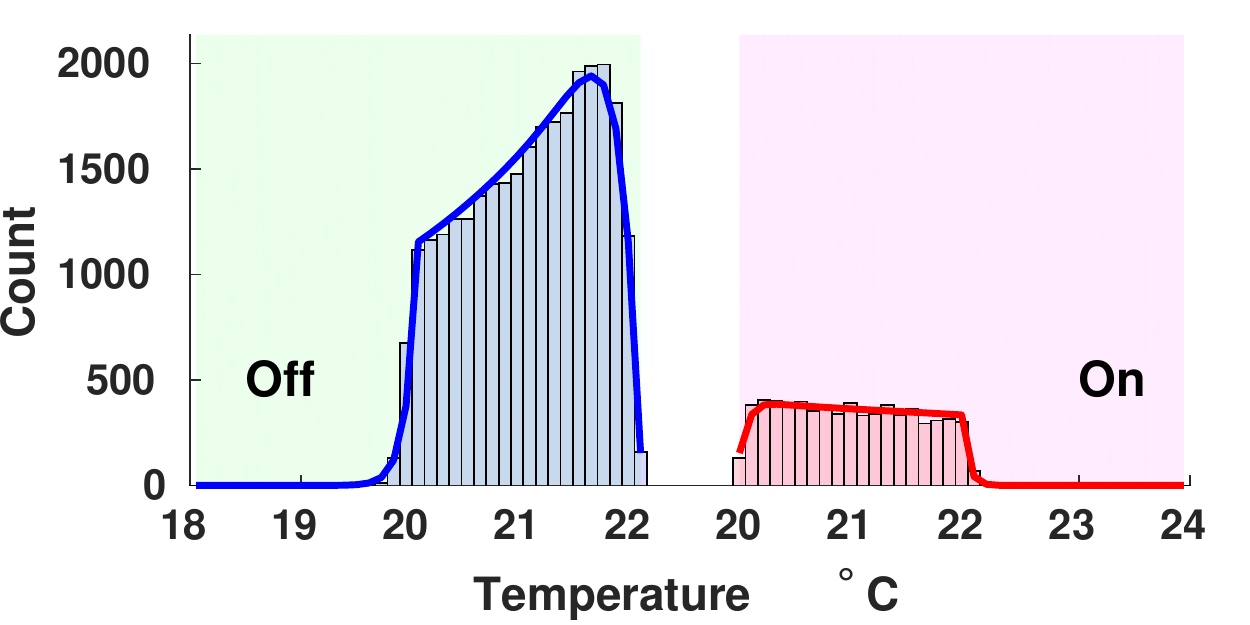}
	\caption{Histogram of the ensemble $H_k$ compared with the marginals $\nu_{\textOff}[k]$ and $\nu_{\textOn}[k]$ obtained from the aggregate model.}
	\label{fig:histSimCompare}
\end{figure}

\begin{figure}
	\centering
	\includegraphics[width = 1\columnwidth]{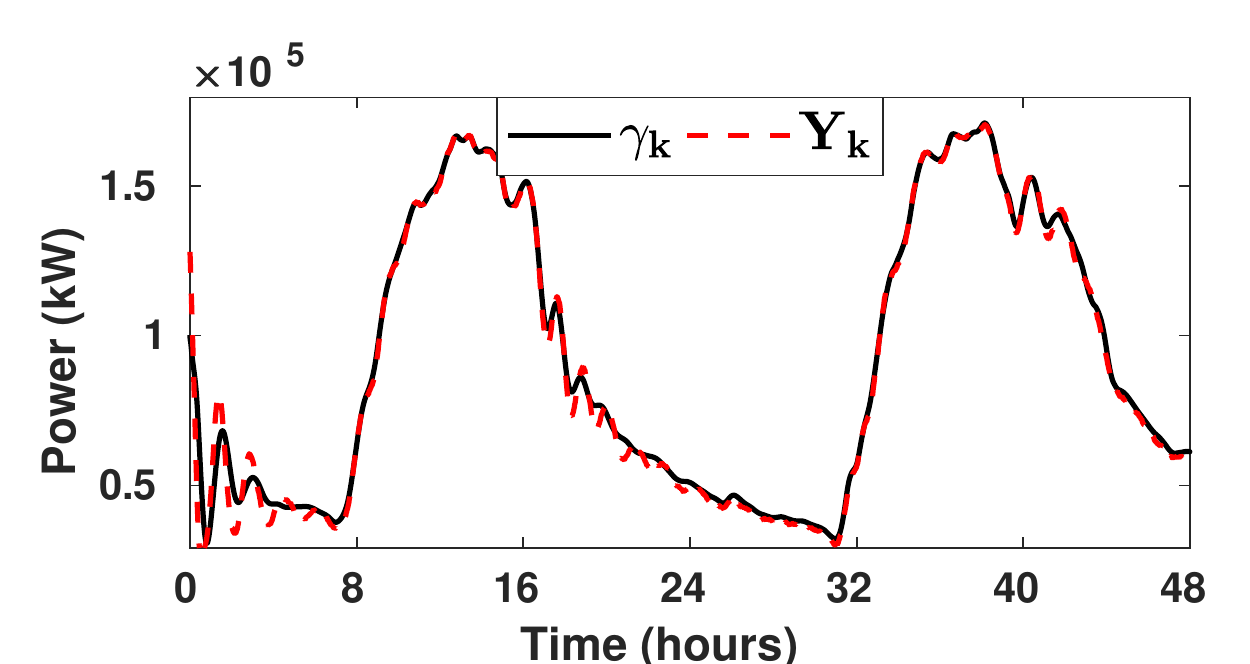}
	\caption{The total power consumption $Y_k$ compared with the output of the aggregate model, $\gamma_k$.}
	\label{fig:powerSimCompare}
\end{figure}

\begin{figure}
	\centering
	\includegraphics[width=1\columnwidth]{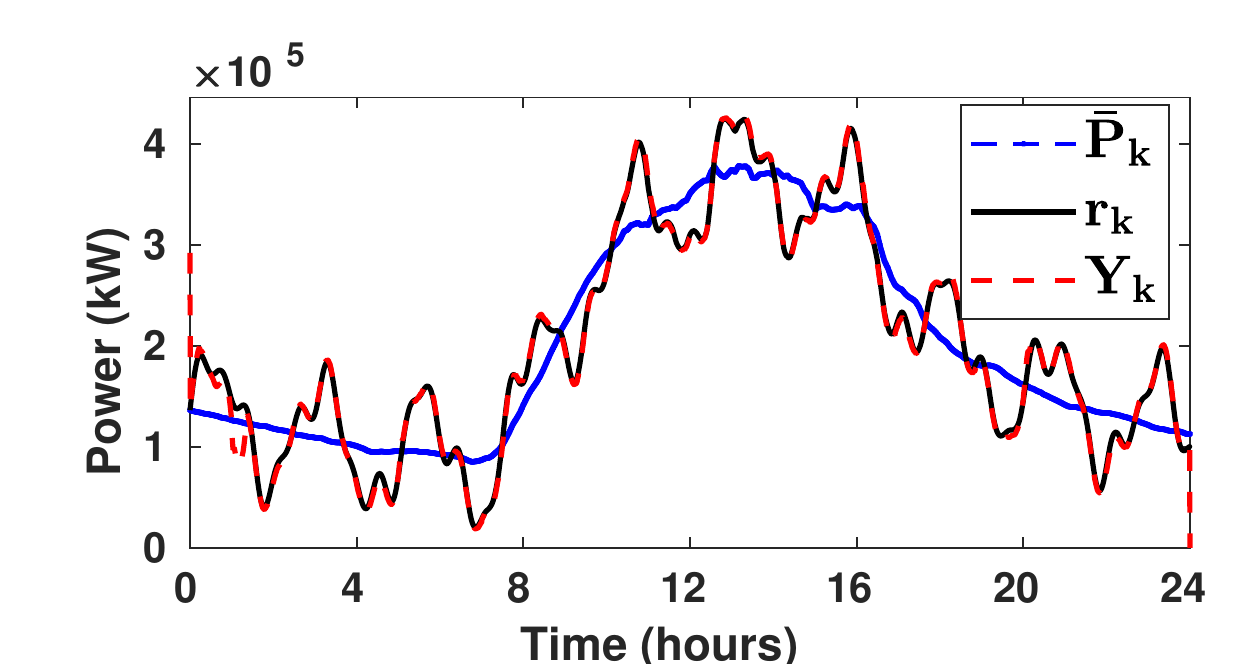}
	\caption{The power consumption of the ensemble $Y_k$ compared with the desired reference $r_k$ and the baseline power $\bar{P}_k$.}
	\label{fig:powerTrackCent}
\end{figure}

\ifx 0
\begin{figure}
	\centering
	\includegraphics[width = 1\columnwidth]{bassPowerCompare.pdf}
	\caption{AA}
	\label{fig:bassPowerSimCompare}
\end{figure}

\begin{figure}
	\centering
	\includegraphics[width = 1\columnwidth]{natOnSwitchCompare.pdf}
	\caption{AA}
	\label{fig:natOnSwitchCompare}
\end{figure}

\begin{figure}
	\centering
	\includegraphics[width = 1\columnwidth]{natOffSwitchCompare.pdf}
	\caption{AA}
	\label{fig:natOffSwitchCompare}
\end{figure}
\fi

\section{Conclusion} \label{sec:conc}
We discretize the Fokker-Planck equations, derived in the past literature~\cite{MalhameElectricTAC:1985}, for a population of TCLs. The discretized equations are then shown to satisfy a certain factorization: the effects of weather and control factor out. The discretized model is verified in simulation, and preliminary results of using the model with its identified factorization for control are shown. Future work entails incorporating the cycling state into the obtained model.

\ifx 0
\section{Resource Allocation}
Resource allocation refers to the allocation of a resource to provide a service. In this context the resource is an ensemble of TCLs, and the service is for the ensembles power consumption to track a desired regulation signal. In this sense, resource allocation is modeled through the following optimization problem,
\begin{align} \label{eq:nonConvOpt}
	\eta^* = \min_{\nu_k,\Phi_k} \ &\eta(\hat{\nu}) = \sum_{k=j}^{j+T}(r_k - \gamma_k)^2  \\
	\text{s.t.} \ &\nu_{k+1} = \nu_k\Phi_kG_k, \ \nu_j = \hat{\nu} \\
				&\nu_k \in [0,1], \ \Phi_k \in \varPhi,
\end{align}
where $r_k$ is the regulation signal, i.e., the desired power consumption at time $k$ and $\varPhi$ is the set defined as
\begin{align} \nonumber
	\varPhi \triangleq \Big\{\Phi \in [0,1]\ \big\vert \ &\phi_{\textOff} \ \text{satisfies}~\eqref{eq:randPolOff2On}, \ \phi_{\textOn} \ \text{satisfies}~\eqref{eq:randPolOn2Off} \Big\}.
\end{align}
Essentially, the set $\varPhi$ ensures that the policies satisfy the pre-specified structure set in~\eqref{eq:randPolOff2On} and~\eqref{eq:randPolOn2Off}.
The problem~\eqref{eq:nonConvOpt} is a non-convex optimization problem, and a well known convexification remedy~\cite{ManneLinearINFORMS:1960} is to consider the following change of variables,
\begin{align}
	J_k = \text{diag}(\nu_k)\Phi_k,
\end{align}
where each non-zero element of $J_k$ can be thought of as the joint distribution $\Prob\left(m_{k+1} = v,\ I_{k+1} = i, \ m_k = u\right)$.
By construction, we have that
\begin{align}
	\nu_{k+1} = \mathbb{1}^TJ_kG_k, \quad \text{and} \quad \nu^T_k = J_k\mathbb{1},
\end{align}
since $\mathbb{1}^T\text{diag}(\nu_k) = \nu_k$  and $\mathbb{1} = \Phi_k\mathbb{1}$. However, due to the structure of our policy, and consequently $\Phi_k$, we do not need to declare the entirety of $J_k$ as a decision variable. For instance, we see that $\text{diag}(\nu_k)\Phi_k$ is a block matrix, where each matrix block is a diagonal matrix. We express this as: $\text{diag}(\nu_k)\Phi_k =$
\begin{align} \label{eq:blockPolMat}
&\begin{bmatrix}
B_{\textOff,\textOff}[k] & B_{\textOff,\textOn}[k] & \mathbf{0} & \mathbf{0} \\
\mathbf{0} & \mathbf{0} & B_{\textOn,\textOff}[k] & B_{\textOn,\textOn}[k]
\end{bmatrix}	\\ \nonumber &\delequalRHS\text{sparse}(J_k),
\end{align}
where, e.g., $B_{\textOff,\textOff}[k] = \text{diag}(\nu_{\textOff}[k])(I - \Phi_{\textOff}[k])$. The other diagonal matrices appearing in~\eqref{eq:blockPolMat} can be inferred by carrying out the matrix multiplication. Additionally, the equality constraints in $\varPhi$ are all of the form $\phi_{\textOff}(\textOn \ \vert \ j) = \kappa$ and can be enforced in our new decision variables through the definition of the joint distribution as 
\begin{align}
\Prob\left(m_{k} = \text{on},\ I_{k} = j, \ m_k = \text{off}\right) = \kappa\nu_\textOff[\lambda^j,k].
\end{align}
Requiring both the joint distribution and marginal distribution to be within $[0,1]$ with the constraint $\nu^T_k = J_k\mathbb{1}$ enforces the inequality constraints in $\varPhi$.
We also have found it necessary to include constraints of the form,
\begin{align} \label{eq:constAddOne}
	&\phi_{\textOff}(\text{on} \ \vert \ j-1)\nu_\textOff[\lambda^{j-1},k] \leq \phi_{\textOff}(\text{on} \ \vert \ j)\nu_\textOff[\lambda^{j},k] \\ \label{eq:constAddTwo}
	&\phi_{\textOn}(\text{off} \ \vert \ j+1)\nu_\textOn[\lambda^{j+1},k] \leq \phi_{\textOn}(\text{off} \ \vert \ j)\nu_\textOn[\lambda^j,k]
\end{align}
so to suggest that the switching on (resp., switching off) probability increases as temperature increases (resp., decreases).
We denote the transcription of the constraints in $\varPhi$ and the additional constraints~\eqref{eq:constAddOne} and~\eqref{eq:constAddTwo} in the new decision variables as $\text{sparse}(J_k) \in \text{sparse}(\varPhi)$. Now, defining the vector: $z_k \triangleq$
\begin{align}
	[\nu_{k+1}, B_{\textOn,\textOff}[k], B_{\textOff,\textOn}[k],B_{\textOn,\textOn}[k],B_{\textOff,\textOff}[k]],
\end{align}
we can formulate an optimization problem in terms of $z_k$ that is equivalent to~\eqref{eq:nonConvOpt} but is now convex, 
\begin{align} \label{prob:convOpt}
	\eta^* = \min_{\{z_k\}_{k=j}^{j+T}} \ &\eta(\hat{\nu}) = \sum_{k=j}^{j+T}(r_k - \gamma_k)^2 \\ \nonumber
\text{s.t.} \ &\forall \ k \in \{j,\dots, j+T\} \\ &\nu_{k+1} = \mathbb{1}^T \text{sparse}(J_k)G_k, \\
		&\nu^T_k = \text{sparse}(J_k)\mathbb{1}\\
		& z_k \in [0,1], \ \nu_j = \hat{\nu}, \\
		&\text{sparse}(J_k) \in \text{sparse}(\varPhi),
\end{align}
where $T>0$ is the time horizon. If $J_k$ was declared directly as a decision variable the problem~\eqref{prob:convOpt} would have $(8N^2+2N)T$ primal variables, whereas the problem with $z_k$ as a decision variable only has $6NT$ primal variables. As an example, suppose that $N=30$ and $T = 280$ the problem~\eqref{prob:convOpt} without the structure exploited has $\approx 2$ million decision variables, utilizing the structure reduces the dimensionality to $\approx 50000$; \emph{two} orders of magnitude reduction. 

We denote the optimal values to~\eqref{prob:convOpt} at index $k$ as $z^*_k = [\nu^*_{k+1}, B^*_{\text{on},\text{off}}[k], B^*_{\text{off},\text{on}}[k],B^*_{\text{on},\text{on}}[k],B^*_{\text{off},\text{off}}[k]]$. Based on the construction of sparse($J_k$), the optimal randomized policies at index $k$ are obtained from the elements of $z^*_k$ (and $z^*_{k-1}$) as,
\begin{align} \label{eq:optPolOff2On}
\phi^*_{\textOff}[k]&= \Big(\text{diag}(\nu^*_{\textOff}[k])^{\dagger}B_{\textOff,\textOn}^{*}[k]\Big)\mathbb{1}, \\ \label{eq:optPolOn2Off}
\phi^*_{\textOn}[k]&= \Big(\text{diag}(\nu^*_{\textOn}[k])^{\dagger}B_{\textOn,\textOff}^{*}[k]\Big)\mathbb{1},
\end{align}
where for a diagonal matrix $A$ the $i^{th}$ diagonal element of $A^{\dagger}$ is
\begin{align}
[A^{\dagger}]_{i,i} = \begin{cases}
1/[A]_{i,i}, & [A]_{i,i} \neq 0. \\
0, & [A]_{i,i} = 0.
\end{cases}
\end{align}
\subsection{Hierarchical solution} \label{sec:hierSolnNoCyc}
We solve the problem~\eqref{prob:convOpt} on one centralized computer, then broadcast the string of policies defined by~\eqref{eq:optPolOff2On}-\eqref{eq:optPolOn2Off} to each TCL. In this example, this is done once and the ensemble of TCLs operate in open loop for $T=24$ hours. Once each TCL receives the string of policies, mode state decisions are made each sampling time with the policies (recall, a description for how exactly this is done in Remark~\ref{rem:impRandPol}). 

The reference signal $r_k$ in this experiment is an arbitrarily generated sum of sinusoids signal added to the baseline power, $\bar{P}^{\textAgg}_k$. The ambient air temperature is time varying and is obtained from  weatherunderground.com for a typical summer day in Gainesville, Fl.

The power consumption of an ensemble of TCLs all making mode state decisions according to the broadcasted policies is shown in Figure~\ref{fig:powerTrackCent}. In Figure~\ref{fig:histTrackCent}, we plot the histogram of the ensemble at a single time instance with the marginals $\nu_{\textOff}[k]$ and $\nu_{\textOn}[k]$ obtained from the aggregate model. The main take away is that the aggregate model: (i) accurately captures the distribution of the population in open loop and (ii) allows for the design of policies for TCLs to track regulation signals in open loop. In practice, the loop can be closed and the problem~\eqref{prob:convOpt} can be solved in Model Predictive Control (MPC) fashion for improved robustness to uncertainty.
\begin{figure}
	\centering
	\includegraphics[width=1\columnwidth]{histTrackingCentralized.pdf}
	\caption{Comparison of the marginals $\nu_{\textOff}[k]$ and $\nu_{\textOn}[k]$ with the histogram of the population $H_k$.}
	\label{fig:histTrackCent}
\end{figure}
\begin{figure}
	\centering
	\includegraphics[width=1\columnwidth]{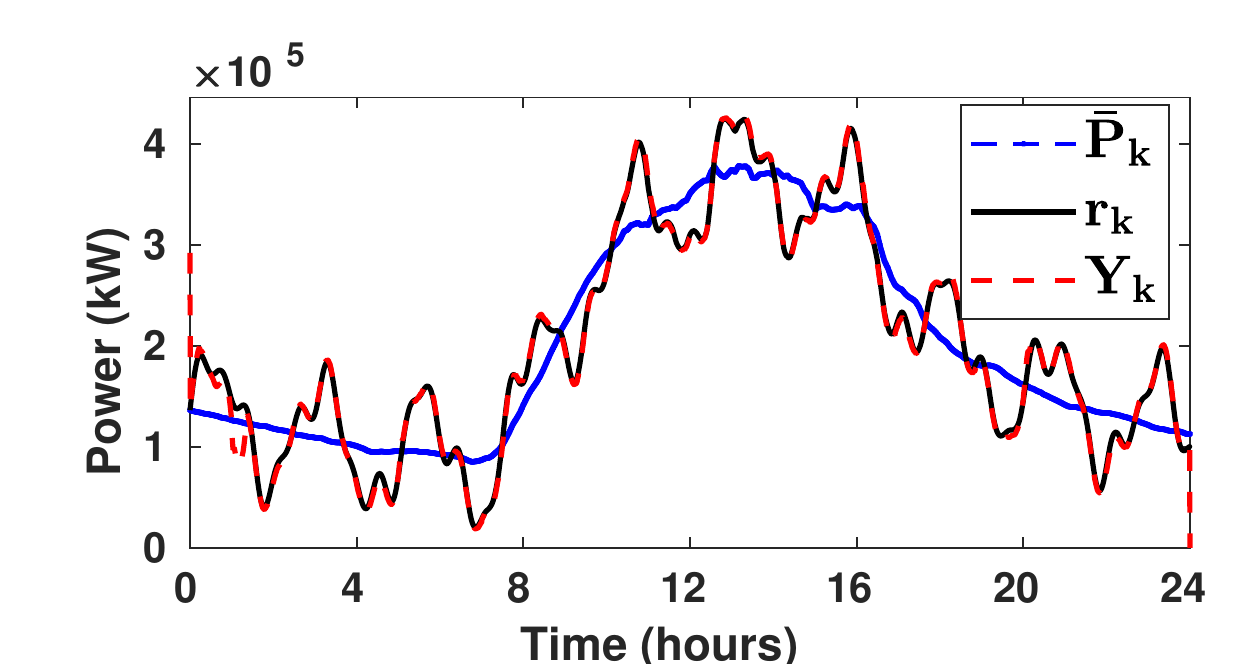}
	\caption{The power consumption of the ensemble $Y_k$ compared with the desired reference $r_k$ and the baseline power $\bar{P}^{\textAgg}_k$.}
	\label{fig:powerTrackCent}
\end{figure}

\section{Incorporating cycling constraints}
In the preceding section it was shown that the problem~\eqref{prob:convOpt} can be efficiently solved to yield a policy that an actual TCL can implement to track the desired regulation signal. What was not shown was the potential capability of extreme short cycling of the TCLs when they follow the said policy. Fortunately, through state augmentation it is possible to incorporate the lockout effect~\cite{LiuDistributedTIE:2016}. We describe how to include the cycling constraint as an additional constraint, and provide an analytical formula for the transition matrix with this in mind. Once this transition matrix is determined it can be used in the optimization problem~\eqref{prob:convOpt} to solve for policies that consider the lockout effect.
\subsection{Individual Cycling QoS Model}
First, we denote an on/off mode switch as,
\begin{align}
S^{\textOn,\ell}_{k-1} \triangleq \begin{cases}
1, & \text{if} \ (m^{\ell}_{k}-m^{\ell}_{k-1}) = 1.\\
0, & \text{otherwise.}
\end{cases} \\	
S^{\textOff,\ell}_{k-1} \triangleq \begin{cases}
1, & \text{if} \ (m^{\ell}_{k-1}-m^{\ell}_{k}) = 1.\\
0, & \text{otherwise.}
\end{cases}
\end{align}
Since an on or off switch are mutually exclusive with each other, a switch occurring is the sum: $S^{\ell}_k=S^{\textOn,\ell}_{k} + S^{\textOff,\ell}_{k}$. We then represent the TCLs locked on status as $\LockedEll_k$, which is expressed as:
\begin{align}
	\LockedEll_k \triangleq \begin{cases}
	1, & \sum_{j=0}^{\tau - 1}S^{\ell}_{k-j} = 1.\\
	0, & \text{otherwise}.
	\end{cases}.
\end{align}
This express that a TCL can switch once in $\tau$ time units. We now define the following,
\begin{align} \label{eq:timeSpentLocked}
	T_k \triangleq \text{Time spent locked since previously not locked.}
\end{align}
By definition~\eqref{eq:timeSpentLocked}, $T_k \in \{0,\dots, \tau\}$ with $T_k = 0$ implying that the TCL is not locked. We can express $T_k$ as an exogenously perturbed inventory system in state feedback,
\begin{align} \label{eq:dynTimeLockedOn}
	T_k = T_{k-1} + \LockedEll_{k-1} - \phi_{d}(T_{k-1}),
\end{align}
where the policy is deterministic and given by
\begin{align}
	\phi_d(T) = \begin{cases}
	\tau, & T = \tau. \\
	0, & \text{otherwise}.
	\end{cases}
\end{align}
The inventory system~\eqref{eq:dynTimeLockedOn} can also be viewed as a controlled Markov chain evolving on the finite state space $\{0,\dots, \tau\}$ with $d_{k-1} = \phi_{d}(T_{k-1})$ the control input. The controlled transition probabilities for this chain are of the form,
\begin{align}
	\Prob(T_{k+1} = \tau_1 \ \vert \ T_k = \tau_2, \ d_k = d, \ \LockedEll_k = l).
\end{align}

We now need to augment the joint input state values $(T_k,d_k)$ to the previous joint input state values $(I_k,m_k)$. While it is possible to take the model~\eqref{eq:dynTimeLockedOn}, form a transition matrix, and then rigorously connect the two state pairs, this construction will most likely offer little insight to the reader. 

\subsection{Combined Transition Matrix}

Rather, we present an intuitive construction based on the following: If a TCL switches from on to off (or off to on), it immediately becomes locked and must remain locked for $\tau$ time instances to satisfy its cycling QoS. With this, we now describe how $T_k$ will evolve under two scenarios: (i) $I_k < N$ and $m_k =$ off or $I_k > 1$ and $m_k =$ on and (ii) $I_k = N$ and $m_k = $off or $I_k = 1$ and $m_k=$ on. 

\subsubsection{Scenario 1} For example, if $m_k = \text{on}$ and $m_{k+1} = \text{off}$ then necessarily $T_k = 0$ and $T_{k+1} = 1$. Further, if for either $m_k =$ on or $m_k =$ off if we have that $T_k = \tau_1$ with $1\leq \tau_1< \tau$ then necessarily $T_{k+1} = \tau_1 + 1$. Lastly, if $T_k = \tau$ then necessarily $T_{k+1} = 0$. Hence, the formation of the transition matrix with state space $(I_k,m_k,T_k)$ can be done by enforcing state transitions to respect the evolution for $T_k$ just described.

\subsubsection{Scenario 2} In this scenario, we must specify which constraint has preference, the temperature constraint or the cycling constraint. That is, at a time $k$ it could be the case that $T_k = \tau_1$ with $\tau_1 < \tau$, $I_k = N$, and $m_k =$ off. In this case, $m_{k+1} =$ on is required to ensure the temperature constraint is satisfied but doing so will violate the cycling QoS. We elect the temperature constraint to have higher preference in this scenario, and reset the counter to $T_{k+1} = 1$. Hence, the device becomes locked in the on state at time $k+1$. The same argument is identical for when $I_k = 1$ and $m_k = $ on.

The sparsity pattern for the transition matrix with state $(I_k,m_k,T_k)$ is shown in Figure~\ref{fig:sparPatternCycling}. An expression for this matrix is, 
	\begin{align}
		\begin{bmatrix}
		e^1\otimes (I-\Phi_{\textOff})P^{\textOn} &\vline&  e^2\otimes \Phi_{\textOff}S^{\textOn} \\
		C_\tau \otimes (I - \Phi^{\textTS}_{\textOff})P^{\textOn} &\vline& (\mathbf{1}\otimes e^2)\otimes \Phi^{\textTS}_{\textOff}S^{\textOn} \\
		\hline 
		e^1\otimes (I-\Phi_{\textOn})P^{\textOff} & \vline & e^2\otimes \Phi_{\textOn}S^{\textOff} \\
		C_\tau \otimes P^{\textOff} &\vline& (\mathbf{1}\otimes e^2)\otimes \Phi^{\textTS}_{\textOn}S^{\textOff}
		\end{bmatrix},
	\end{align}
where $e^\ell$ be the $\ell^{th}$ canonical basis vector in $\mathbb{R}^{\tau+1}$
	
	\subsection{Hierarchical Solution with cycling state}
	Similar to section~\ref{sec:hierSolnNoCyc} we solve the optimization problem~\eqref{prob:convOpt}, but now with including the additional state $T_k$. The reference tracking results are shown in Figure~\ref{fig:refTrackPlot_withCycState} and the cycling QoS is shown in Figure~\ref{fig:cycQos_withCycState}.
	\begin{figure}
		\centering
		\includegraphics[width=1\columnwidth]{histCyclingQoS_withCycleState.pdf}
		\caption{Histogram of all the times between successive mode state switches for the population of TCLs. The vertical line is centered at the value $\tau = 10$ on the x-axis.}
		\label{fig:cycQos_withCycState}
	\end{figure}

	\begin{figure}
		\centering
		\includegraphics[width=1\columnwidth]{powerTrackingCentralized_withCycleState.pdf}
		\caption{The power consumption of the ensemble $Y_k$ compared with the desired reference $r_k$ and the baseline power $\bar{P}^{\textAgg}_k$ for the system with the cycling state including.}
		\label{fig:refTrackPlot_withCycState}
	\end{figure}	

\section{Conclusion}
We design a convex optimization problem for the control of TCLs.
\fi
\ifshowArxivAlt
\bibliographystyle{IEEEtran}
\bibliography{/home/austin/allBibFiles/Barooah,/home/austin/allBibFiles/optimization,/home/austin/allBibFiles/ControlTheory,/home/austin/allBibFiles/basics,/home/austin/allBibFiles/distributed_control,/home/austin/allBibFiles/grid}	
\fi

\ifshowArxiv

\fi

\end{document}